\documentclass[sigconf,10pt]{acmart}
\def\fullversion{0}

\usepackage{caption}

\PassOptionsToPackage{table}{xcolor}
\usepackage{colortbl}
\usepackage{style}
\usepackage{bigstrut}
\usepackage{microtype}
\usepackage{tabularx}
\usepackage{multirow}
\usepackage{multicol}
\usepackage{rotating}
\usepackage{lipsum}
\usepackage{balance}
  \newcommand{\dacrecur}{{$(\alpha,\beta,k,l,m)$-recurrence}\xspace}

  \newcommand{\op}[1]{\texttt{#1}}

  \newcommand{\forkins}{\texttt{fork}\xspace}
  \newcommand{\joinins}{\texttt{join}\xspace}

  \newcommand{\thread}{thread\xspace}

  \newcommand{\depth}{span\xspace}

  \newcommand{\tas}{{\texttt{test\_and\_set}}\xspace}
  
  \newcommand{\cas}{{\texttt{compare\_and\_swap}}\xspace}

  \newcommand{\kdgrid}{$k$-d grid\xspace}

  \newcommand{\ifconference}{{{\ifx\fullversion\undefined}}}

\hide{
\newtheoremstyle{exampstyle}
  {.05in} % Space above
  {.05in} % Space below
  {} % Body font
  {.5em} % Indent amount
  {\sc \bfseries} % Theorem head font
  {.} % Punctuation after theorem head
  {.5em} % Space after theorem head
  {} % Theorem head spec (can be left empty, meaning `normal')
\theoremstyle{exampstyle} 
\theoremstyle{exampstyle} 

\makeatletter
\renewenvironment{proof}[1][\proofname]{\par
  \vspace{-\topsep}% remove the space after the theorem
  \pushQED{\qed}%
  \normalfont
  \topsep0pt \partopsep0pt % no space before
  \trivlist
  \item[\hskip\labelsep
        \itshape
    #1\@addpunct{.}]\ignorespaces
}{%
  \popQED\endtrivlist\@endpefalse
  \addvspace{3pt plus 3pt} % some space after
}
}

\usepackage{array}
\acmConference[APoCS'22]{
Symposium on Algorithmic Principles of Computer Systems (APOCS22)
}{
January 12, 2022
}{
Alexandria, Virginia
}
\acmYear{2022}

\begin{document}
\fancyhead{}
\title[]{Analysis of Work-Stealing and \\Parallel Cache Complexity}
  \author{Yan Gu}
  \affiliation{\institution{UC Riverside}}
  \email{ygu@cs.ucr.edu}
  \author{Zachary Napier}
  \affiliation{\institution{UC Riverside}}
  \email{znapi001@ucr.edu}
  \author{Yihan Sun}
  \affiliation{\institution{UC Riverside}}
  \email{yihans@cs.ucr.edu}
\begin{abstract}
Parallelism has become extremely popular over the past decade, and there have been a lot of new parallel algorithms and software.
The randomized work-stealing (RWS) scheduler plays a crucial role in this ecosystem.
In this paper, we study two important topics related to the randomized work-stealing scheduler.

Our first contribution is a simplified, classroom-ready version of analysis for the RWS scheduler.
The theoretical efficiency of the RWS scheduler has been analyzed for a variety of settings, but most of them are quite complicated.
In this paper, we show a new analysis, which we believe is easy to understand, and can be especially useful in education.
We avoid using the potential function in the analysis, and we assume a highly asynchronous setting, which is more realistic for today's parallel machines.

Our second and main contribution is some new parallel cache complexity for algorithms using the RWS scheduler.
Although the sequential I/O model has been well-studied over the past decades, so far very few results have extended it to the parallel setting.
The parallel cache bounds of many existing algorithms are affected by a polynomial of the span, which causes a significant overhead for high-span algorithms.
Our new analysis decouples the span from the analysis of the parallel cache complexity.
This allows us to show new parallel cache bounds for a list of classic algorithms.
Our results are only a polylogarithmic factor off the lower bounds, and significantly improve previous results.
\end{abstract} 

\maketitle
\makeatletter
\newcommand{\removelatexerror}{\let\@latex@error\@gobble}
\makeatother

%\newpage

\section{Introduction}
Hardware advances in the last decade have brought multicore parallel machines to the mainstream.
While there are multiple programming paradigms and tools to enable parallelism in multicore machines, the one based on \defn{nested parallelism} with \defn{randomized work-stealing (RWS) scheduler} is with no doubt the most popular and widely used.
The nested parallelism model and its variants have been supported by most parallel programming languages (e.g., Cilk, TBB, TPL, X10, Java Fork-join, and OpenMP), introduced in textbooks (e.g., Cormen, Leiserson, Rivest and Stein~\cite{CLRS}), and employed in a variety of research papers (to list a few: \cite{agrawal2014batching,blelloch2010low,BCGRCK08,BG04,Blelloch1998,blelloch1999pipelining,BlellochFiGi11,BST12,Cole17,BBFGGMS16,dinh2016extending,chowdhury2017provably,blelloch2018geometry,dhulipala2020semi,BBFGGMS18,Dhulipala2018,blelloch2020randomized,gu2021parallel,blelloch2016justjoin,sun2018pam,sun2019parallel,ptreedb}).
At a high level, this model allows an algorithm to recursively and dynamically create (\emph{fork}) parallel tasks, which will be executed on $P$ processors by a dynamic scheduler.
This nested (binary) fork-join provides a good abstraction for shared-memory parallelism.
On the user (algorithm designer or programmer) side, it is a simple extension to the classic programming model with additional keywords for creating new tasks (e.g., \op{fork}) and synchronization (e.g., \op{join}) between tasks.
%On the system side, the randomized work-stealing scheduler automatically runs any such computations on all available processors, and also provides good theoretical guarantees.
%This abstraction is also used in most of the modern textbooks and courses for parallel algorithms and programming (to list a few: XXX).

The \defn{randomized work-stealing (RWS) scheduler} plays a crucial role in this ecosystem. It automatically and dynamically maps a nested parallel algorithm to the hardware efficiently both in theory and in practice. In this paper, we study two important topics related to the RWS scheduler. First, we show a \defn{new, simplified, and classroom-ready version of analysis for the RWS scheduler}, which we believe is easier to understand than existing ones.
Second, we provide some \defn{new analyses for parallel cache complexity} for nested-parallel algorithms based on the RWS scheduler.
We provide a list of almost optimal bounds listed in \Cref{table:intro}, and more discussions will be given later.

Our first contribution is a simplified analysis for the RWS scheduler. The theoretical efficiency of the RWS scheduler was first given by Blumofe and Leiserson~\cite{blumofe1999scheduling}, and is later analyzed for a variety of settings (to list a few: \cite{Acar02,itpa,muller2016latency,acar2013scheduling,BBFGGMS16,singer2019proactive,Acar2016tapp,singer2020scheduling,arora2001thread}).
Although these analyses essentially consider more complex settings (e.g., to also consider external I/Os), the analyses themselves are quite complicated.
To the best of our knowledge, the details of the proofs are covered in very few courses related to parallelism, and in most cases, RWS is just treated as a black box.
%Hence, we consider the opposite: for the simple case (e.g., no complicated future operations or external I/Os in the algorithms), what the simplest analysis can be, or the most comprehensible in education.
Hence, we simply consider the goal of bounding the number of steals of the RWS scheduler, and we want to answer the question of what the simplest analysis for the RWS scheduler can be, or what is the most comprehensible version in education.

This paper presents a simplified analysis in \Cref{sec:rws}.
Unlike most of the existing analyses, our version does not rely on defining the potential function for a substructure of the computation, which we believe is easy to understand.
%\footnote{We note that this will limit the applicability of this analysis in some settings, but the goal of our work is to provide a simple version for the plain setting.}.
Our analysis is inspired by a recent analysis~\cite{itpa} that is similar to~\cite{agrawal2008adaptive,suksompong2014bounds}.
Our analysis differs from~\cite{itpa} in two aspects.
First, unlike~\cite{itpa}, we do not assume all processors run in lock-steps (the PRAM setting).
Instead, we assume a highly asynchronous setting, which is more realistic for today's parallel machines.
Second, we separate the math calculation from the details of the RWS algorithm, which may be helpful for classroom teaching.

Our second and main contribution of this paper is on the parallel cache complexity for algorithms using the RWS scheduler.
On today's machines, the memory access cost usually dominates the running time of most combinatorial algorithms.
To capture this, sequentially, Aggarwal and Vitter~\cite{AggarwalV88} first formalized the external-memory model to capture the I/O cost of an algorithm, which was refined by Frigo et al.~\cite{Frigo99} as the ideal-cache model.
The cost measure is called \emph{I/O complexity}~\cite{AggarwalV88} (noted as $Q_1$) or \emph{cache complexity}~\cite{Frigo99} when specifying the communication cost between the cache and the main memory.
While this model has received great success in the algorithm and database communities, so far, few results have extended it to the parallel (distributed cache) setting, which we summarize in \Cref{sec:cachebound}.
Among them, Acar, Blelloch and Blumofe~\cite{Acar02} first defined and showed that the parallel cache complexity $Q_P$ is at most $Q_1+O(PDM/B)$, where $P$, $D$, $M$ and $B$ are the number of processors, span (aka.\ depth, the longest critical path of dependences), cache size, and cache block size, respectively (definitions in \Cref{sec:prelim}).
However, this bound is pessimistic and usually too loose. %and counterintuitive. %: for the same algorithm, this $Q_p$ bound increases as the increase of $M$ (cache size), but in reality it is the opposite.
Frigo and Strumpen~\cite{frigo2009cache} and later work by Cole and Ramachandran~\cite{cole2013analysis,cole2012revisiting} showed tighter parallel cache complexity for many cache-oblivious algorithms summarized \Cref{table:intro}.
%carefully analyzed the scheduling algorithm and the computational structure for certain cache-oblivious algorithms, and showed tighter parallel cache complexity as compared to~\cite{Acar02}.
However, the bounds have a polynomial overhead\footnote{Such an overhead is usually $n^{a}$, where $a$ is a constant and $1/2\le a<1$.} on the input size when the span is polynomial to the input size.
Such overhead can easily dominate the cache bound when plugging in the real-world input size, as discussed in the caption of \Cref{table:intro}.
It remained an open question on how to close this gap.

In this paper, we significantly close this gap for a variety of algorithms.
The polynomial overhead in the previous analysis~\cite{Acar02,frigo2009cache,cole2013analysis,cole2012revisiting} is due to the high span of these algorithms (a polynomial of the algorithm's span shows up in the bound).
The key insight in our analysis is to break such polynomial correlation between the algorithm's span and the overhead for parallel cache complexity.
Our new analysis is inspired by the abstraction of \kdgrid{} proposed recently by Blelloch and Gu~\cite{BG2020}, which was previously used to show sequential cache bounds.
We extend the idea for analyzing parallel cache complexity.
%which was originally used to show lower bounds and sequential cache bounds independently with the parallel dependence structure of the computation.
Our analysis directly studies the ``real'' dependencies in the computation structure of these algorithms, instead of those caused by parallelism.
In particular, the core of our analysis is just the \emph{recurrences} of these algorithms.
As a result, the parallel dependency of the computation (the algorithm's span) does not show up in the analysis.
This helps us avoid the complication of digging into the details of the scheduling algorithms in the analysis, and the analysis is no more complicated than just solving recurrences.
To do this, we define the \dacrecur{}, which effectively models the pattern of recurrences of many classic algorithms, and show a general theorem, \Cref{thm:recur}, to solve these recurrences. By applying these results, we can achieve the new bounds in \Cref{table:intro}.
%By using $k$d-grid and the analysis introduced in \Cref{sec:pcc}, we achieve the new bounds in \Cref{table:intro}.
The parallel overheads are polylogarithmic instead of polynomial to the input size, compared to the lower bounds by~\cite{ballard2014communication,BG2020}.
% It is worth mentioning that all the algorithms are existing, and only the analyses are new.
We believe the methodology is of independent interest, and could be useful in analyzing other parallel or sequential algorithms. We leave this as future work.
%We note that unlike the previous work~\cite{frigo2009cache,cole2013analysis,cole2012revisiting} that analyzed how the computation is scheduled by RWS, we directly study the recurrence structures of the computations.

%In summary, the contributions of this paper include:
The contributions of this paper are as follows.

\begin{enumerate}
  \item We show a new analysis for the randomized work-stealing scheduler, which avoids the use of potential functions and is very simple.
  \item We propose the \dacrecur{} and a general theorem (\Cref{thm:recur}) for solving it, which applies to solving the parallel cache complexity for many classic parallel algorithms.
  \item We show new cache complexity bounds for a variety of algorithms, shown in \Cref{table:intro}, which significantly improves existing results, and is very close to the lower bounds (only a polylogarithmic overhead).
\end{enumerate}

%//Check if we want to mention lower-order terms later.

\begin{table*}
  \def\arraystretch{1.4}
  \small
  \begin{tabular}{@{}p{4cm}<{\centering}@{  }c@{  }@{  }c@{  }@{  }c@{}c@{}}
  \toprule
    \multirow{2}[0]{*}{Algorithm} & Seq. & \multicolumn{3}{c}{Parallel Overhead}   \\\vspace{-1.5em}
    & Bound & New in this paper & Previous best~\cite{frigo2009cache,cole2012revisiting,BG2020} & Lower Bound~\cite{ballard2014communication,BG2020}\\
  \midrule
%Matrix Multiplication & $\frac{n^3}{B\sqrt{M}}+\frac{n^2}{B}+1$ & - & \(P^{1/3}\log^{2/3}{n}\cdot\frac{n^2}{B} + P\log^2 n\) & $P^{1/3}\cdot\frac{n^2}{B}$\\
Gaussian Elimination  & \multirow{2}[0]{*}{$\frac{n^3}{B\sqrt{M}}+\frac{n^2}{B}+1$} & \multirow{2}[0]{*}{\(P^{1/3}\log^{2/3}{n}\cdot\frac{n^2}{B} + Pn\)} & \multirow{2}[0]{*}{\(P^{1/3}{n^{1/3}}\cdot\frac{n^2}{B} + Pn\log B\)}& \multirow{2}[0]{*}{$P^{1/3}\cdot\frac{n^2}{B}$}\\\vspace{-1.5em}
Kleene's algorithm for APSP\\

Triangular System Solver & $\frac{n^3}{B\sqrt{M}}+\frac{n^2}{B}+1$ & \(P^{1/3}\log^{5/3}n\cdot\frac{n^2}{B} + Pn\) & \(P^{1/3}n^{1/3}\cdot\frac{n^2}{B} + Pn\log B\)& $P^{1/3}\cdot\frac{n^2}{B}$\\

Cholesky Factorization & \multirow{2}[0]{*}{$\frac{n^3}{B\sqrt{M}}+\frac{n^2}{B}+1$} & \multirow{2}[0]{*}{$P^{1/3}\log^{5/3}n\cdot\frac{n^2}{B}+Pn\log n$}& \multirow{2}[0]{*}{$P^{1/3}n^{1/3}\log^{1/3}n\cdot\frac{n^2}{B}+P n \log n$} &\multirow{2}[0]{*}{$P^{1/3}\cdot\frac{n^2}{B}$}\\\vspace{-1.5em}
LU Decomposition & \\

LWS Recurrence & $\frac{n^2}{BM}+\frac{n}{B}+1$ & \(P^{1/2} \log^2{n}\cdot\frac{n}{B} + Pn\) & \(P^{1/2}\cdot n^{1/2}\cdot\frac{n}{B}+Pn\) & $P^{1/2}\cdot\frac{n}{B}$\\
GAP Recurrence & $\frac{n^3}{BM}+\frac{n^2\log M}{B}+1$ & \(P^{1/2} \log^2{n}\cdot\frac{n^2}{B} + P n^{\kappa}\) & \(P^{1/2}n^{\kappa/2}\log{M}\cdot\frac{n^2}{B} + Pn^{\kappa}\)& $P^{1/2}\cdot\frac{n^2}{B}$\\
Parenthesis Recurrence & $\frac{n^3}{B\sqrt{M}}+\frac{n^2}{B}+1$ & \(P^{1/3} \log^{5/3}{n}\cdot\frac{n^2}{B} + Pn^{\kappa}\) & \(P^{1/3}n^{\kappa/3}\cdot\frac{n^2}{B} + Pn^{\kappa}\) & $P^{1/3}\cdot\frac{n^2}{B}$\\
RNA Recurrence & $\frac{n^4}{BM}+\frac{n^2}{B}+1$ & \(P^{1/2} \log^2{n}\cdot\frac{n^2}{B} + Pn^{\kappa}\) & \(P^{1/2} n^{\kappa/2} \cdot\frac{n^2}{B}+Pn^{\kappa}\) & $P^{1/2}\cdot\frac{n^2}{B}$\\
Protein Accordion Folding & $\frac{n^3}{BM}+\frac{n^2}{B}+1$ & \(P^{1/2}\log n \cdot \frac{n^2}{B} + P n \log^2 n\) & \(P^{1/2} n^{1/2} \log n \cdot \frac{n^2}{B} + Pn\log^2n\) & $P^{1/2}\cdot\frac{n^2}{B}$\\
\bottomrule
\end{tabular}
\caption{\textbf{Sequential and parallel cache complexity for a list of algorithms.}  Here $M$, $B$, and $P$ are the cache size, the block size, and the number of processors, respectively, and $\kappa=\log_23\approx 1.58$.
The sequential bounds are from various existing papers~\cite{BG2020,Frigo99,chowdhury2008cache,dinh2016extending}.
The parallel cache complexity is the sequential cache complexity plus the parallel overhead.
The best previous parallel bounds are either from~\cite{frigo2009cache,cole2012revisiting}, or what we computed based on the algorithms from~\cite{BG2020} and the analysis from~\cite{frigo2009cache,cole2012revisiting}.
The new bounds are analyzed in Sec.~\ref{sec:pcc}.
For all algorithms/problems, the new parallel bounds are much tighter.
Using LWS recurrence as an example, for the best previous bound, the parallel overhead always dominates unless \(P^{1/2}n^{3/2}/{B}<n^2/BM\), which gives $n=\omega(PM^2)$ (at least $2^{45}$ when plugging in real-world values for $M$ and $P$, which is impossible to store and compute).
Our new bound improve the dominating term of $O(P^{1/2}n^{3/2}/B)$ by $\sqrt{n}/\log^2n$.
Similarly, for other algorithms, our new bounds replace the $n^a$ in the dominating terms for $a=1/2$, $1/3$, $\kappa/2$, or $\kappa/3$, with the $\log^b n$ for $b\le 2$.
In all cases, the gaps between new parallel upper bounds and the lower bounds on the specific computations are $\log^b n$ for $b\le 2$, which were previously a polynomial factor $n^a$ for $a=1/2$, $1/3$, $\kappa/2$, or $\kappa/3$.
\label{table:intro}}

\end{table*}

\hide{
\section{old}
Hardware advances in the last decade have brought multicore parallel machines to the mainstream.
%Nowadays, it is hard to find a single-core processor.
While there are multiple programming paradigms and tools to enable parallelism in multicore machines, the one based on \emph{nested fork-join} with \emph{randomized work-stealing scheduler} has no doubt to be the most popular and widely-used. %is the model for parallel algorithms in Cormen, Leiserson, Rivest and Stein~\cite{CLRS}.
In a high level, an algorithm can recursively and dynamically create (\emph{fork}) parallel tasks, which will be executed and executed on all cores. We give more details of the model in \Cref{sec:prelim}.

Nested (binary) fork-join provides a good abstraction for shared-memory parallelism.
On the user (algorithm designer or programmer) side, it is a simple extension to the classic programming model with additional keywords for fork and join.
On the system side, the randomized work-stealing scheduler automatically runs the code on all cores, with good theoretical guarantees.
This abstraction is also used in most of the modern textbooks and courses for parallel algorithms and programming (to list a few: XXX).

The randomized work-stealing (RWS) scheduler plays an crucial role in this ecosystem since it dynamically maps a parallel algorithm to the hardware in an efficient manner both theoretically and practically.
The theoretical efficiency of the RWS scheduler has been first shown by Blumofe and Leiserson~\cite{BL98}, and later analyzed for a variety of settings (to list a few: YYY).
Unfortunately, all of these analyses are quite complicated, and the RWS scheduler is usually treated as a black box.
To the best of our knowledge, the details of the proofs are not covered in any of the existing courses listed above.
Given the importance of the RWS scheduler, it is of crucial relevance to simplify the analysis that is more comprehensible in courses.

This paper shows such a simplified analysis.
Unlike most of the existing analyses (BL and others), our version does not rely on defining the potential function for a substructure of the computation, which we believe is easier to understand\footnote{We note that this will limit the applicability of this analysis in more settings, but the goal here is to provide a simple version for the plain setting.}.
Our analysis is inspired by a recent analysis in~\cite{itpa}.
However, unlike that in~\cite{itpa} (a single long proof), our proof breaks down into several lemmas that each is simple both conceptually and in details.
Then, the proof of the theorem simply combines the lemmas.
We believe this version is particularly suitable in the classroom.

Our second and main contribution of this paper is about parallel cache complexity based on the RWS scheduler.
Nowadays, memory access cost dominates the running time of most combinatorial algorithms.
To capture this, sequentially, Aggarwal and Vitter~\cite{AggarwalV88} first formalized the external-memory model to capture the I/O cost of an algorithm, which was refined by Frigo et al.~\cite{Frigo99} as the ideal-cache model.
The cost measure is I/O complexity~\cite{AggarwalV88} (usually $Q_1$) or cache complexity~\cite{Frigo99} when specifying the communication cost between the cache and the main memory.
While this model has received a great success in the algorithm and database communities, so far we have few results on extending it to the parallel (distributed cache) setting.
The two main results include the follows.
Acar, Blelloch and Blumofe~\cite{Acar02} showed that the parallel cache complexity $Q_p$ is $Q_1+O(pDM)$, where $p$, $D$ and $M$ are the number of cores, span (aka.\ depth, the longest critical
path of dependences), and cache size.
However, this bound is usually too loose and counterintuitive: for the same algorithm, $Q_p$ increases as the increase of $M$ (cache size), but in reality it is the opposite.
Hence, Frigo and Strumpen~\cite{frigo2009cache} and later work by Cole and Ramachandran~\cite{cole2013analysis,cole2012revisiting} analyzed the RWS scheduler and the computational structure for certain cache-oblivious algorithms, and showed tighter parallel cache complexity as compared to~\cite{Acar02}.
However, their bounds are still relatively loose, and have an overhead that is polynomial on the input size when the span is polynomial to the input size.
More details can be found in \Cref{tab:mainresult}, and the parallelism overhead can easily dominate the cache bound when plugging in real-world input size.

In this paper, we show how to break this polynomial dependence between the algorithm's span and the overhead for parallel cache complexity.
We note that unlike the previous work~\cite{frigo2009cache,cole2013analysis,cole2012revisiting} that analyzed how the computation is scheduled by RWS, we directly study the recurrence structures of the computations.
This novel idea will lead to a list of tighter bounds as shown in \Cref{tab:mainresult}.
This is inspired by the abstraction of \kdgrid proposed recently by Blelloch and Gu~\cite{BG2020}, which was originally used to show cache lower bounds independently with the dependence structure of the computation.
Compared to the lower bounds by~\cite{ballard2014communication,BG2020}, the overheads of all bounds shown in this paper are polylogarithmic instead of polynomial to the input size.

//Check if we want to mention lower-order terms later.
}

\section{Preliminaries}
\label{sec:prelim}

\subsection{Nested Parallelism}\label{sec:np-prelim}
The {nested parallelism model} is a programming model for shared-memory parallel algorithms.
This model allows algorithms to recursively and dynamically create new parallel tasks (threads).
The computation will be simulated (scheduled) on $P$ loosely synchronized processors, and explicit synchronization can be used to
let threads reach consensus. Some commonly-used examples include the (binary) fork-join model and binary forking model.
%Here we use the definition from~\cite{blelloch2020optimal}.
More precisely, in this model, we assume a set of \thread{}s that have
access to a shared memory.
A computation starts with a single {root} \thread{} and finishes when all \thread{s} finish.
Each \thread{} supports the same
operations as in the sequential RAM model, but also has a \forkins{}
instruction that forks a new child \thread{} that can be run in parallel.
%Once the child \thread{}s finish, they use the \joinins{} operation to merge back and continue execution.
Threads can synchronize with each other using some primitives. The most common synchronization is \joinins{}, where every \forkins{} corresponds to a later \joinins{}, and the \forkins{} and corresponding \joinins{} are properly nested. In more general models (e.g., the binary-forking model~\cite{blelloch2020optimal}) threads are also allowed to be synchronized with any other threads, probably by using atomic primitives such as \tas{} or \cas{}.
%The \joinins{} can be nested (i.e., each \joinins{} corresponds to a \forkins{}, referred to as the binary fork-join model) or non-nested (referred to as the binary-forking model).
%The nested version is more restrictive, but both models have the same theoretical guarantees.
This model is the most widely-used model for multicore programming, and is supported by many parallel languages including NESL~\cite{blelloch1992nesl},
Cilk~\cite{frigo1998implementation}, the Java fork-join framework~\cite{Java-fork-join}, OpenMP~\cite{OpenMP}, X10~\cite{charles2005x10}, Habanero~\cite{budimlic2011design},
Intel Threading Building Blocks~\cite{TBB}, the Task Parallel Library~\cite{TPL}, and many others.
In this model, we usually require the computation to be either \defn{race-free}~\cite{feng1999efficient} (i.e., no logically parallel instructions access the same memory
location and at least one is a write), or to only use atomic operations (e.g., \tas{} or \cas{}) to deal with concurrent writes (e.g.,~\cite{blelloch2012internally}).

\subsection{Work-Span Measure}

%The work-span measure applies to a nested-parallel computation (aka.\ a series-parallel computation DAG) as follows.
For a computation using nested parallelism, we can measure its \emph{work} and \emph{span} by evaluating its series-parallel computational DAG (i.e., a DAG modeling the dependence between operations in the computation). The \defn{work} $W$ is the number of operations in this computation, or the costs of all tasks in the computation DAG (the time complexity on the RAM model). The \defn{span} (or \defn{depth}) $D$ is the maximum number of operations over all directed paths in the computation DAG.  %This measure is also called as critical path length or the \defn{depth}.
%\begin{itemize}
%  \item The \defn{work} $W$ is the number of operations in this computation, or the costs of all tasks in the computation DAG (the time complexity on the RAM model).
%  \item The \defn{span} $D$ of is the maximum number of operations over all directed paths in the computation DAG.  This measure is also called as critical path length or the \defn{depth}.
%\end{itemize}

\subsection{Randomized Work-Stealing (RWS) Scheduler}\label{sec:rws-prelim}

In practice, a nested-parallel computation can be scheduled on multicore machines using the \defn{randomized work-stealing (RWS)} algorithm.
More details of the RWS scheduler can be found in~\cite{blumofe1999scheduling}, and here we overview the high-level ideas that will be used in our analysis.
The RWS scheduler assigns one double-ended queue (\defn{deque}) for each processor that can execute a \thread{} at a time.
One processor starts with taking the root \thread{}.
Each processor then proceeds as follows:
\begin{itemize}
  \item If the current \thread{} runs a \forkins, the processor enqueues one \thread{} at the front of its queue (spawned child), and executes the other \thread{} (continuation).
  \item If the current \thread{} completes, the processor pulls a \thread{} from the front of its own queue.
  \item If a processor's queue is empty, it randomly selects one of the other processors, and steals a \thread{} from the end of that processor's queue (victim queue). If that fails, the processor retries until succeeds.
\end{itemize}
Since a steal can be pretty costly (involves complicated inter-processor communication) in practice, a common practice is to wait for at least the time for a successful steal before retrying~\cite{Acar02}.

The overhead of executing the computation based on the RWS scheduler is mainly on the steal attempts (both successful ones and failed ones), plus maintaining the deque.
Hence, bounding the number of steals is of great interest for multicore parallelism.

\begin{theorem}\label{thm:rws}
  Executing a series-parallel computation DAG with work $W$ and span $D$ on an RWS scheduler uses $O(PD)$ steals \whp{} to $W$,\footnote{We use the term $O(f(n))$ with high probability (\whp{}) in $n$ to indicate the bound $O(kf(n))$ holds with probability at least $1-1/n^k$ for any $k\ge 1$. With clear context we drop ``in $n$''.} where $P$ is the number of processors.
\end{theorem}
\Cref{thm:rws} was first shown by Blumofe and Leiserson~\cite{blumofe1999scheduling}, and was later proved in different papers (e.g.,\cite{Acar02,itpa,muller2016latency,acar2013scheduling,BBFGGMS16,singer2019proactive,Acar2016tapp,singer2020scheduling,arora2001thread}).
Since RWS is an asynchronous algorithm, certain synchronization assumptions are required.
Early work~\cite{blumofe1999scheduling} assumed all processors are fully synchronized and all operations have unit cost.
Later work~\cite{Acar02} relaxed it (and thus made it more realistic) that a steal attempt takes at least $s$ and at most $ks$ time steps where $s$ is the cost for a steal and $k\ge 1$ is a constant.
In this paper, we further relax the assumption of synchronization---we assume that, between a failed steal and the next steal attempt of each processor, every other processor can try at most $k$ steal attempts for some constant $k\ge 1$.
Also, between two steal attempts from the same processor, another processor that has work to do will execute at least one instruction.
We believe such a relaxation is crucial since, in practice, processors are highly asynchronous due to various reasons, including cache misses, processor pipelines, branch prediction, hyper-threading, changing clock speeds, interrupts, or the operating system scheduler.
Hence, it is hard to define what time steps mean for different processors.
However, it is reasonable to assume that processors run in similar speeds within a constant factor, and the steal attempts are not too often (in practice the gap between two steal attempts is usually set to be at least hundreds of to thousands of cycles).
In the analysis, we consider the simpler case for $k=1$, but it is easy to see that a larger $k$ will not asymptotically affect the scheduling result (\Cref{thm:rws}) as long as $k$ is a constant.

\subsection{Cache Complexity}\label{sec:cachebound}

Cache complexity (aka.\ I/O complexity) measures the memory access cost of an algorithm, which in many cases can be the bottleneck of the execution time, especially for parallel combinatorial algorithms.
The idea was first introduced by Aggarwal and Vitter~\cite{AggarwalV88}, and has been widely studied since then.
Here we use definition by Frigo et al.~\cite{Frigo99} that is the most adopted now.
Here we assume a two-level memory hierarchy.
The CPU is connected to a small-memory (cache) of size $M$, and this small-memory is connected to a large-memory (main memory) of effectively infinite size.
Both small-memory and large-memory are divided into blocks of size $B$ (cachelines), so there are $M/B$ cachelines in the cache.
The CPU can only access the memory on blocks resident in the cache and it is free of charge.
Finally, we assume an optimal offline cache replacement policy, which is automatic, to transfer the data between the cache and the main memory, and a unit cost for each cacheline load and evict.
The practical policy such as LRU or FIFO, are $O(1)$-competitive with the optimal offline algorithm if they have a cache with twice the size.
The cache complexity of an algorithm, $Q_1$, is the total cost to execute this algorithm on such a model.

The above measure is sequential.
One way to extend it to the parallel setting is the ``Parallel External Memory'' (PEM) model~\cite{Arge08} that analogs the PRAM model~\cite{shiloach1981finding}.
However, since modern processors are highly asynchronous~\cite{blelloch2020optimal} instead of running in lock-steps as assumed in PRAM, the PEM model cannot measure the communication between cache and main memory well.
An alternative solution is to assume multiple processors work independently, and they either share a common cache or own their individual caches.
For individual caches (aka.\ distributed caches), the parallel cache complexity $Q_P$ based on RWS is upper bounded by $Q_1+O(SM/B)$ \whp, where $S$ is the number of steals~\cite{Acar02}.
This is easy to see since each steal can lead to, at worst, an entire reload of the cache, with cost $O(M/B)$.
Applying \Cref{thm:rws}, we can get $Q_P\le Q_1+O(PDM/B)$.

There have also been studies on other cache configurations, such as shared caches~\cite{BG04}, multi-level hierarchical caches~\cite{BlellochFiGi11,SBFGK14,BCGRCK08}, and varying cache sizes~\cite{bender2014cache,bender2016cache}.
Many parallel cache-efficient algorithms have been designed inspired by these measurements (e.g.,~\cite{cole2010resource,tang2011pochoir,BCGRCK08,chowdhury2008cache,blelloch2010low,chowdhury2010cache,BG2020,BFGGS15,Arge08}).

\section{Simplified RWS Analysis}\label{sec:rws}
%\red{need to reword everything about ``between two steals''}
The randomized work-stealing (RWS) scheduler plays a crucial role in the ecosystem of nested-parallel algorithms and multicore platforms, and it dynamically maps the algorithms to the hardware. RWS is efficient both theoretically and practically.
The theoretical efficiency of the RWS scheduler has been first given by Blumofe and Leiserson~\cite{blumofe1999scheduling}. A lot of later work analyzed RWS in different settings~\cite{Acar02,itpa,muller2016latency,acar2013scheduling,BBFGGMS16,singer2019proactive,Acar2016tapp,singer2020scheduling,arora2001thread}.
As mentioned, most of these analyses are quite involved, especially when they also consider some more complicated settings (e.g., external I/O costs). This makes the analysis very hard to cover in undergraduate or graduate courses, and RWS is usually treated as a black box.
Given the importance of the RWS scheduler, it is crucial to make the analysis more comprehensible so that it can be taught in classes.
%Hence, we consider the opposite: for the simplest case (no caching, future, external I/O, etc.), what the simplest analysis can be, or the most comprehensible in education.
In the following, we will show a simplified analysis for RWS, which proves \Cref{thm:rws}.

Unlike most of the existing analyses, our version does not rely on defining the potential function for a substructure of the computation.
We understand that this will limit the applicability of this analysis in some settings (e.g., we do not extend the results to consider external I/O costs), but our goal is to provide a simple proof for a reasonably general setting (the binary-forking model~\cite{blelloch2020optimal}), and make it easy to understand in most parallel algorithm courses.

Our analysis is inspired by a recent analysis in~\cite{itpa} that is similar to those in~\cite{agrawal2008adaptive,suksompong2014bounds}.
Our analysis differs from~\cite{itpa} in two aspects.
First, the analysis in~\cite{itpa} assumes all processors run in lock-steps (PRAM setting).
However, on today's machines, the processors are loosely synchronized with different relative processing rates changing over time.
Hence, the processors can run in different speeds and do not and should not be synchronized.
%In our analysis, we assume between two consecutive steal attempts on each processor, every other processor can execute at least one operation, and can make at most a constant number of steal attempts.
In our analysis, we assume, between a failed steal attempt and the next steal attempt on each processor, every other processor can make at most $k$ steal attempts for some constant $k\ge 1$.
We believe this is a reasonably realistic assumption that all today's machines and RWS implementations satisfy.
For simplicity, in our analysis we use $k=1$, and it is easy to see that this does not affect the asymptotical bounds.
%In the rest of our analysis, for simplicity, we assume that between two consecutive steal attempts on each processor, each other processor can only make at most one steal attempt, and apparently this does not only affect the bound asymptotically.

Secondly, instead of showing a single long proof, we tried our best to improve its understandability, and separate the math and calculation from the main idea of the proof.
We start from the most optimistic case as a motivating example in \Cref{thm:ws-lemma-1}, and then generalize it to \Cref{thm:ws-lemma-2} and \Cref{thm:ws-lemma-3}.
%Then for \Cref{thm:ws-lemma-2} and the proof of \Cref{thm:rws}, which include more details about RWS, the math parts are straightforward---either in the same form as the motivating example, or a simple union bound.
The only mathematics tools we use are Chernoff bound and union bound.
These lemmas pave the path to, and decouple the mathematics from, the proof of \Cref{thm:rws},
which involves more details and the core idea about RWS. %, the math parts are straightforward---either in the same form as the motivating example, or a simple union bound.
We believe this can be helpful for classroom teaching.

\medskip
We say two steal attempts \defn{overlap} with each other when they choose the same victim processor, and happen concurrently.
To start with, we first show the simplest case where none of the steal attempts overlap with each other.
This is the optimal case because in fact multiple attempts may happen simultaneously and only one can succeed.
We show the lemma below as a starting point of our analysis.
%We start with this case for warm-up purpose.

\begin{lemma}\label{thm:ws-lemma-1}
  %It takes $(p-1)(D + \log(1/\epsilon))$ non-overlapping steal attempts to steal at least $D$ tasks from the deque of one specific processor with probability at least $1-\epsilon$.
  Given a victim processor $\Pi$ and $(P-1)(D + \log(1/\epsilon))$ non-overlapping steal attempts, the probability that at least $D$ tasks from the deque of the processor $\Pi$ are stolen is at least $1-\epsilon$, where $P$ is the number of processors.
\end{lemma}
Here we overload the notation $D$ in the lemma that was previously defined as the span of the computation.
We do so because later in the proof of \Cref{thm:rws}, we plug in $D$ as the span of the algorithm, so we just use $D$ for convenience.
We use the classic version of Chernoff bound that considers independent random variables $X_1, \ldots, X_t$ taking values in $\{0, 1\}$.
Let $X$ be the sum of these random variables, and let $\mu = \E[X]$ be the expected value of $X$, then for any ``offset'' $0<\delta<1$, $\Pr(X\le(1-\delta)\mu)\le e^{-\delta^2\mu/2}$.

\begin{proof}[\Cref{thm:ws-lemma-1}]
  Recall that in RWS, each steal attempt independently chooses a victim processor and tries to steal a task.
  Hence, it has probability $1/(P-1)$ to choose processor $\Pi$ (and steal one task if so).
  We now consider each of the steal attempts as a random variable $X_i$. $X_i=1$ if it chooses processor $\Pi$, and $0$ otherwise.
  In Chernoff Bound, let the number of random variables $t=2(P-1)(D + \log(1/\epsilon))$.
  Then for $t$ steals, the expected number of hits is $\mu=2(D + \log(1/\epsilon))$.
  We are interested in the probability $p$ that fewer than $D$ steal attempts hit processor $\Pi$.
  In this case, $\delta=(\mu-D)/\mu$, so $(1-\delta)\mu=D$.
  Applying Chernoff bound, we can get that the probability $p$ is no more than $e^{-\delta^2\mu/2}$.
  Here $\delta^2\mu/2=(\mu-D)^2/2\mu>(\mu-2D)/2=\ln(1/\epsilon)$.
  The ``$>$'' step is because we discard the $D^2/2\mu$ term that is always positive.
  Hence, $e^{-\delta^2\mu/2}<e^{-\ln(1/\epsilon)}=\epsilon$, which proves the lemma.
\end{proof}

Next, we consider the general case that multiple processors try to steal concurrently, so the steal attempts can overlap.
In this case, although unlikely, many processors may make the steal attempts ``almost'' at the same time, where only one processor wins (using arbitrary tie-break) and the others fail.
As mentioned, we assume that between a failed steal attempt and the next steal attempts on one processor, every other processor can have at most one steal attempt.
\begin{lemma}\label{thm:ws-lemma-2}
  %It takes $(p-1)e(D + \log(1/\epsilon))/(e-1)$ steal attempts from $p-1$ processors to steal at least $D$ tasks from the deque of one specific processor with probability at least $1-\epsilon$.
  Given a specific victim processor $\Pi$ and $(P-1)e(D + \log(1/\epsilon))/(e-1)$ steal attempts from $P-1$ other processors, the probability that $D$ tasks from the deque of the processor $\Pi$ are stolen is at least $1-\epsilon$.
\end{lemma}
\begin{proof}
  Although multiple processors can attempt to steal concurrently, two steals from the same processor will never be concurrent.
  %, since they cannot happen at the same time.
  %overlap with another steal attempt from another processor simultaneously.
  %Hence, the most pessimistic situation is that $p-1$ steals from $p-1$ processors always happen together, which maximizes the chance that a steal hits the queue but fails to get the task since it is taken by another simultaneous steal.
  Hence, based on our assumption, the most pessimistic situation is that the $P-1$ processors always have $P-1$ steals at the same time, which maximizes the chance that a steal hits the queue but fails to get the task. %since it is taken by another simultaneous steal.
  The probability that at least one of the $P-1$ concurrent steals chooses this victim processor $\Pi$ (so that at least one of the tasks is stolen) is at least $1-(1-1/(P-1))^{P-1}>1-1/e$.

  We can similarly use Chernoff bound to show that the probability that fewer than $D$ tasks are stolen after $S'=2(D+\log(1/\epsilon))/(1-1/e)$ steps is small.
  We consider each random variable as a group of $P-1$ steal attempts, with probability of at least $1-1/e$ to choose the deque of the specific processor $\Pi$.
  In this case, the expected value of the sum $\mu=2(D + \log(1/\epsilon))$ and the offset $\delta=(\mu-D)/\mu$ remain the same as in \Cref{thm:ws-lemma-1}, so the probability is also the same.
\end{proof}

Here note that in the analysis, we do not need the assumption that the $D$ tasks are in the same deque of a certain processor.
In fact, the analysis easily extends to when the $D$ tasks are from different processors as long as there is always one task available to be stolen.
Therefore, we show the relaxed form of \Cref{thm:ws-lemma-2}.
\begin{lemma}\label{thm:ws-lemma-3}
  %It takes $(p-1)e(D + \log(1/\epsilon))/(e-1)$ steal attempts to steal at least $D$ tasks from different processor with probability at least $1-\epsilon$, as long as one task is always available at the time of any steal attempt.
  Given $D$ tasks and $(P-1)e(D + \log(1/\epsilon))/(e-1)$ steal attempts, the probability that these $D$ tasks are stolen is at least $1-\epsilon$, as long as at least one task is available at the time of any steal attempt.
\end{lemma}

As discussed in \Cref{sec:np-prelim}, any nested-parallel computation can be viewed as a DAG, and each (non-termination) node is an instruction and has either two successors (for a \forkins{}) or one successor (otherwise).
The RWS scheduler dynamically maps each node to a processor.
For each specific path, the length is no more than the span $D$ of the algorithm, based on the definition.
Based on the RWS algorithm, a node will be mapped to the same processor that executes the predecessor node, except for the spawned children (definition in \Cref{sec:rws-prelim}).
The spawned child of a processor $\Pi$ is ready to be stolen during the process when $\Pi$ is executing the other branch (continuation), and will be executed by $\Pi$ if it is not stolen during this process.

We now prove Theorem~\ref{thm:rws} by showing that the computation must have been terminated after $O(PD)$ steals.
We will use Lemma~\ref{thm:ws-lemma-3} and apply union bound.

\begin{proof}[Theorem~\ref{thm:rws}]
  We consider a path in the DAG and show that all instructions on this path will be executed with no more than $O(PD)$ steals.
  Each node on this path is either a spawned child (that can be stolen) or executed directly after the previous node by the same processor.
  Now let's consider a processor that is not working on the instructions on this path.
  When the next steal is attempted, the processor working on this path either has added one more node $v$ on this path that is ready to be stolen, or has executed the node $v$.
  This is because we assume a processor executes at least one instruction between two steal attempts from another processor.
  The only case that the next node is not executed is when it is a spawned child.
  It will not be executed immediately, and needs to wait until to be stolen for execution, or for the continuation branches to finish and execute these nodes.

  Hence, let's consider the worst case that all nodes on the path are spawned children.
  Lemma~\ref{thm:ws-lemma-3} upper bounds the number of steals to finish the execution of this path.
  Namely, after $(P-1)e(D + \log(1/\epsilon))/(e-1)$ steal attempts, all nodes are stolen and executed with probability at least $1-\epsilon$.
  For a DAG with the longest path length $D$, there are at most $2^D$ paths in the DAG.
  Now we set $\epsilon=1/(2^D\cdot W^c)$ where $W$ is the work of the computation (the number of nodes in the DAG).
  For any constant $c\ge 1$, $(P-1)e(D + \log(2^D\cdot W^c))/(e-1)=O(P(D+\log W))$ steals are sufficient for executing all existing paths.
  Now we take the union bound on the probability that all $2^D$ paths will finish, which is $1-2^D\cdot\epsilon=1-W^{-c}$.
  Since each node in the DAG can have at most two successors, the DAG needs to have $D=\Omega(\log W)$ longest path length to contains $W$ nodes.
  Hence, the $\log W$ term will not dominate, which simplifies the number of steals to be $O(PD)$.
\end{proof}

%We have tried to cover this analysis in a number of parallel algorithm courses,
We have attempted to include this analysis in a few lectures of parallel algorithm courses, and we also would like to include the answer to a frequently asked question.
The question is, in the analysis, we apply union bound on $2^D$ paths, but apparently, the $O(PD)$ steals cannot cover all paths since $PD$ is a much lower-order term than $2^D$ in practice.
Theoretically, the answer is that $\epsilon=2^D\cdot n^c$ is a sufficiently small term for us to apply union bound, which can give us the desired bound in \Cref{thm:rws}.
The more practical and easy-to-understand answer is that we are assuming the worst case, and in practice we do not need all spawn children to be stolen in the execution.
In fact, it is likely that most of them are executed by the same processor that spawns this child.
Take a parallel-for-loop as an example, which can be viewed $\log n$ level of binary-forks.
For most of the paths in this DAG, the spawn children are executed by the same processor that executes the parent node.
Because of the design of the RWS algorithm, most of the successful steals will involve a large chunk of work, so steal attempts are infrequent.
The analysis shows that, once $O(PD)$ steals are made, the path must have finished, but it is more likely that the computation has finished even before this number of attempts are made.

\section{Analysis for Parallel Cache Complexity}\label{sec:pcc}

Studying parallel cache complexity for nested-parallel algorithms scheduled by RWS is a crucial topic for parallel computing and has been studied in many existing papers (e.g.,~\cite{Acar02,frigo2009cache,cole2012revisiting,cole2013analysis}).
The goal in these analyses is to show the parallel overhead when scheduling using RWS, in addition to the sequential cache complexity.
We show the best existing parallel bounds for a list of widely used algorithms and problems in \Cref{table:intro}.
While the results in these papers are reasonably good for algorithms with low (polylogarithmic) span, the bounds for parallel overhead can be significant for algorithms with linear or super-linear span.
Compared to the lower bounds for the parallel overhead, the upper bounds given in these papers incur polynomial (usually $n^{1/2}$ or $n^{1/3}$) overheads.
Such parallel overhead will dominate most of the input range when compared to the sequential cache bounds.
Meanwhile, it is known that the practical performance of many of these algorithms is almost as good as low-span algorithms~\cite{chowdhury2010cache,SchardlThesis}.
Hence, it remains an open problem for decades to tightly bound the parallel overhead of such algorithms.

In this section, we show a new analysis to give almost tight parallel cache bounds for the list of problems in \Cref{table:intro}, which are only a polylogarithmic factor off the lower bounds for the main term.
Unlike the previous approaches that analyze the scheduler, we directly study the recurrence relations of such computations and find it surprisingly simple.
This new analysis is inspired by the concept of \emph{$k$d-grid}~\cite{BG2020} (see more details in \Cref{sec:approach}).
%This new analysis is surprisingly simple, and later in \Cref{sec:approach} we will show how the new concept of the $k$d-grid~\cite{BG2020} inspired us for this analysis.

In the rest of this section, we first review the existing work on this topic in \Cref{sec:pcc-related}. %, which are also used in our analysis.
%Then we show the high-level idea and the main theorem of our analysis in \Cref{sec:approach}.
\Cref{sec:approach} presents the high-level idea and the main theorem (\Cref{thm:recur}) of our analysis, which provides a general approach to solve the cache complexity based on recurrences.
In \Cref{sec:kleene}, we use a simple example of Kleene's algorithm to show how to use the newly introduced main theorem.
Finally, we show the new results for more complicated algorithms in \Cref{sec:applications}, and discuss the applicability and open problems in \Cref{sec:discuss}.

\subsection{Related Work}\label{sec:pcc-related}

Given the importance of I/O efficiency and the RWS scheduler, parallel cache bounds have been studied for over 20 years.
The definition on distributed cache was given by Acar et al.~\cite{Acar02}, and they also showed a trivial parallel upper bound on $P$ processors: $Q_P\le Q_1+O(PDM/B)$.
To achieve this bound, one just needs to pessimistically assume that in each of the $O(PD)$ steals, the stealing thread accesses the entire cache from the original processor, which is $O(M/B)$ additional cache misses.
This bound is easy to understand and good for algorithms with polylogarithmic span, but is too loose for linear and super-linear span algorithms.
Hence, the following later works showed tighter parallel cache bounds for algorithms with certain structures.

%Frigo and Strumpen~\cite{frigo2009cache} first analyzed a class of computations in a divide-and-conquer fashion with cache complexity for the subproblems as a concave function of the operation cost, such as matrix multiplication and 1d Stencil.
Frigo and Strumpen~\cite{frigo2009cache} first analyzed the parallel cache complexity of a class of divide-and-conquer computations, such as matrix multiplication and 1D Stencil,
where the problems have subproblem cache complexity as a ``concave'' function of the computation cost (see more details in~\cite{frigo2009cache}).
Actually, the idea from~\cite{frigo2009cache} is general and can be applied to a variety of algorithms as shown in \Cref{table:intro}.
Later work by Cole and Ramachandran~\cite{cole2012revisiting} pointed out a missing part of the analysis in~\cite{frigo2009cache}---the additional cache misses by accessing the execution stacks after a successful steal.
They carefully studied this problem, and showed that in most cases, this additional cost is asymptotically bounded by other terms (so the bounds are the same as~\cite{frigo2009cache}). In other cases, this can lead to a small overhead (e.g., matrix multiply in row-major format).
The authors of~\cite{cole2012revisiting} also extended the set of applicable algorithms and showed tighter parallel cache bounds for problems such as FFT and list ranking.
For the algorithms in this paper, we assume the matrices are in bit-interleaved format~\cite{Frigo99}, so algorithms incur no asymptotic cost for accessing the execution stacks after steals.
Even not, we note that all algorithms in \Cref{table:intro} do not require accesses to the cactus stack anyway (many later RWS implementations chose not to support that for better practicality).
%, so in this paper we do not consider such overhead.

In this paper, we do not consider the additional cost of false sharing~\cite{cole2013analysis} or other schedulers~\cite{cole2017bounding,yang2018scheduling}, but it seems possible to extend the analysis in this paper to the other settings. We leave this as future work.

\subsection{Our Approach}\label{sec:approach}

As opposed to directly analyzing the algorithms on the scheduler in previous work~\cite{frigo2009cache,cole2012revisiting,cole2013analysis}, our key observation is to directly study the computation structure of these algorithms.
Interestingly, our analysis is mostly independent with the RWS scheduler, and only plugs in some results from~\cite{frigo2009cache} for some basic primitives such as matrix multiplication.
By doing so, our analysis can bound the parallel cache complexity much better than the previous results.

%The idea of our analysis is motivated by the recent work by Blelloch and Gu~\cite{BG2020}.
%This work abstracts the hidden structure in many parallel dynamic programming and algebra algorithms and refers to it as the $k$d-grid.
%This novel approach decouples the dependencies due to parallelism and the sequential cache complexity, and such correlation was the crux of the previous algorithm design.
The idea of our analysis is motivated by the recent work by Blelloch and Gu~\cite{BG2020}.
This work studies several parallel dynamic programming and algebra problems, and defines a structure called $k$d-grid to reveal the computational structure
of these problems.
It uses $k$d-grid with $k=2$ or $3$ to model a list of classic problems, such as matrix multiplication, to capture the memory access pattern of
these problems.
By using $k$d-grid, their analysis decouples the parallel dependency (and, effectively, the span) from the sequential cache complexity in many parallel algorithms (see details in~\cite{BG2020}).
%By avoiding dealing with this dilemma, a number of lower bounds of the classic algorithms have been shown, as well as the matching upper bounds (new algorithms with better cache complexity).
Although they only applied the $k$d-grid analysis on sequential cache complexity, the high-level idea motivates us to also revisit the analysis of parallel cache complexity, and inspired us to directly analyze the essence of the computation structure (the recurrences) of the algorithms.
%We note that the previous work focused on the sequential case, but we can extend the high-level idea to analyze the parallel cache complexity.
This effectively avoids the crux in previous analysis~\cite{Acar02,frigo2009cache,cole2012revisiting,cole2013analysis}, which incurs a polynomial overhead in the cache complexity charged by the span of the algorithm.
In all of our analyses, the span of these algorithms, no matter linear or super-linear, do not show up in the analysis, which is very different from previous work.
Of course, larger span does lead to more parallel cache overhead since it increases the number of steal attempts and each successful one incurs at least one additional cache miss.
However, in all applications in this paper, this term is bounded by either the sequential bound or the main term for parallel overhead.
Combining all together, we summarize all cache bounds in \Cref{table:intro}, and our new parallel cache bounds for linear and super-linear span algorithms are almost as good as those of the low-span algorithms (e.g., matrix multiplication).

As mentioned, our analysis will use previous results to derive the parallel cache bounds for $k$d-grids, and use the recurrence relation to bound the entire algorithm.
We formalize the recurrences we study for these algorithms and problems, which we refer to as the \dacrecur{}.

\begin{definition}[\dacrecur]
An \dacrecur{} is a recurrence in the following form:
$$Q(n)=\alpha\cdot Q(n/\beta)+\sum{k_i\cdot n^{l_i}\log^{m_i}n}$$
where $l_i$ and $m_i$ are non-negative numbers, and $k_i$ is a function of $P$, $M$ and $B$.
\end{definition}
In the next section, we will use Kleene's algorithm as an example to show an instantiation of this recurrence relation.
The \dacrecur is easy to solve using the master method~\cite{bentley1980general}:
\begin{theorem}[Main Theorem]\label{thm:recur}
  The solution to $Q(n)$, an \dacrecur, is:
  $$O\left(\sum{k_i\cdot n^{l_i}\log^{m_i}n}+\sum{k_j\cdot n^{l_j}\log^{m_j+1}n}+\sum{k_r}n^{\log_\beta\alpha}\right)$$
  for $l_i>\log_\beta\alpha$, $l_j=\log_\beta\alpha$, and $l_r<\log_\beta\alpha$.
\end{theorem}

As shown here and in the next section, the parallel dependencies of the computation (the algorithm's span) do not show up in the analysis and the solution, which is different from the previous analyses~\cite{Acar02,frigo2009cache,cole2012revisiting,cole2013analysis}.

In the rest of this section, we will first use Kleene's algorithm as an example to show how to use our approach to derive tighter parallel cache bound. Then we show a list of cache-oblivious algorithms that we can apply \Cref{thm:recur} to and get improved bounds.

It is worth mentioning that the cache bound contains the term for the call stack of the (recursive) subproblems.
This term is a constant in the sequential bound, and in many cases the parallel term is the same as the number of steals (e.g., for matrix multiplication and Kleene's algorithm).
In other cases, directly applying \Cref{thm:recur} leads to a $O(Pn^{\log_\beta\alpha})$ term, which is suboptimal since Cole and Ramanchadran~\cite{cole2012revisiting} showed that this term can be $O(PD)$.
Hence, when $n^{\log_\beta\alpha}>D$, instead of using $n^{\log_\beta\alpha}>D$ from \Cref{thm:recur}, we plug in the $O(PD)$ term from~\cite{cole2012revisiting}, which gives a tighter result.

\subsection{Kleene's Algorithm as an Example}\label{sec:kleene}

To start with, we use Kleene's algorithm for all-pair shortest-paths (APSP) as an example to explain the analysis.
Kleene's algorithm solves the all-pair shortest-paths (APSP) problem that takes a graph $G=(V,E)$ (with no negative cycles) as input.
The Kleene's algorithm was first mentioned in~\cite{kleene1951representation,munro1971efficient,fischer1971boolean,furman1970application}, and later discussed in full details in~\cite{Aho74}. It is a divide-and-conquer algorithm that is I/O-efficient, cache-oblivious and highly parallelized.
The pseudocode of Kleene's algorithm is in Algorithm~\ref{alg:kleene}.

\begin{algorithm}[!h]
\caption{\mf{Kleene}($A$)}
\label{alg:kleene}
\DontPrintSemicolon
\fontsize{9pt}{11pt}\selectfont
\KwIn{Distance matrix $A$ initialized based on the input graph $G=(V,E)$}
\KwOut{Computed Distance matrix $A$}

\smallskip

\lIf {$|A| = 1$} {\Return{$A$}}

\smallskip

$A_{00}\gets \mf{Kleene}(A_{00})$\\
$A_{01}\gets A_{01}+A_{00}A_{01}$\\
$A_{10}\gets A_{10}+A_{10}A_{00}$\\
$A_{11}\gets A_{11}+A_{10}A_{01}$\\

\smallskip

$A_{11}\gets \mf{Kleene}(A_{11})$\\
$A_{01}\gets A_{01}+A_{01}A_{11}$\\
$A_{10}\gets A_{10}+A_{11}A_{10}$\\
$A_{00}\gets A_{00}+A_{10}A_{01}$\\

\smallskip

\Return{$A$}
\end{algorithm}

In Kleene's algorithm, the graph $G$ is represented as the matrix $A$, where $A[i][j]$ is the weight of the edge between vertices $i$ and $j$ (the weight is $+\infty$ if the edge does not exist).
$A$ is partitioned into 4 submatrices indexed as $\begin{bmatrix}A_{00}&A_{01}\\A_{10}&A_{11}\end{bmatrix}$.
The matrix multiplication is defined in a closed semi-ring with $(+, \min)$.
The high-level idea is first to compute the APSP between the first half of the vertices only using the paths between these vertices.
Then by applying some matrix multiplication, we update the shortest paths between the second half of the vertices using the computed distances from the first half.
We then apply another recursive subtask on the second half vertices.
The computed distances are finalized, and then we use them to update the shortest paths from the first-half vertices.

The cache complexity $Q(n)$ and \depth{} $D(n)$ of this algorithm follow the recurrence relations:
\begin{align}
%\begin{split}
Q(n)=2Q(n/2)+6Q_{\smb{MM}}(n/2)\label{eqn:kleene-cache}\\
%\end{split}
D(n)=2D(n/2)+2D_{\smb{MM}}(n/2)\label{eqn:kleene-span}
\end{align}
where $Q_{\smb{MM}}(n)$ is the I/O cost of a matrix multiplication of input size $n$.
Note that the recurrence relation for the cache complexity is true no matter if we are considering the sequential case (e.g., $Q_1$ and $Q_{\smb{MM},1}$) or the parallel case (e.g., $Q_P$ and $Q_{\smb{MM},P}$).
For the parallel matrix multiplication algorithm from~\cite{BG2020}, we have $D_{\smb{MM}}(n)=O(\log^2 n)$, $D(n)=O(n)$, $\displaystyle Q_{\smb{MM},1}=\Theta\left(\frac{n^3}{B\sqrt{M}}+\frac{n^2}{B}+1\right)$, and $Q_1=\displaystyle \Theta\left(\frac{n^3}{B\sqrt{M}}+\frac{n^2}{B}+n\right)$.
%\yihan{Check the notation $Q,Q_{MM,\cdot},Q_{\cdot}$ in the formula.}

If we directly use the result from~\cite{Acar02}, then we get the parallel cache bound $Q_P=\displaystyle Q_1+O\left(\frac{PnM}{B}\right)$.
As the significant growth of processor count and cache size, the $O({PnM}/{B})$ term dominates unless $n$ is very large.
The tighter bound from~\cite{frigo2009cache} shows $Q_P=\displaystyle Q_1+O\left(\frac{P^{1/3}n^{7/3}}{B}+Pn\right)$.
This bound is tighter than the previous one from~\cite{Acar02}, but the $O(P^{1/3}n^{7/3}/B)$ term still dominates unless $n=\omega(P^{1/2}M^{3/4})$, which is unlikely in practice.
The parallel lower bound for this computation~\cite{ballard2014communication,BG2020} is $\displaystyle Q_P=Q_1+\Omega\left(\frac{P^{1/3}n^2}{B}\right)$, so a polynomial gap remains between the lower and upper cache bounds.
Our analysis significantly closes this gap to polylogarithmic.

Now we use \Cref{thm:recur} to directly solve this \dacrecur.
\Cref{eqn:kleene-cache} includes the cache complexity of matrix multiplication.
The parallel bounds on $P$ processors based on the algorithm from~\cite{BG2020} is:
\begin{equation}\label{eqn:mm-par}
  Q_{\smb{MM},P}=O\left(\frac{n^3}{B\sqrt{M}}+\frac{P^{1/3}n^2\log^{2/3}n}{B}+P\log^{2}n\right).
\end{equation}
which can be shown by the analysis from~\cite{frigo2009cache,cole2012revisiting}.
Now we can plug in \Cref{eqn:mm-par} to \Cref{eqn:kleene-cache}, and get an \dacrecur for $Q_P(n)$.
In this case, we have $\alpha=\beta=2$, and $\{(k_i,l_i,m_i)\}=\{$$(1/B\sqrt{M},3,0),$ $(P^{1/3}/B,2,2/3),$ $(P,0,2)\}$.
Plugging in \Cref{thm:recur} directly gives the solution of:
$$Q_P(n)=O\left(\frac{n^3}{B\sqrt{M}}+\frac{P^{1/3}n^2\log^{2/3}n}{B}+Pn\right).$$

In this case, the input size is $O(n^2)$ so the corresponding term $l_i=2>\log_\beta\alpha=1$.
Hence, even though Kleene's algorithm has linear span as opposed to the polylogarithmic span for matrix multiplication, the additional steals caused by the span will not affect the input term (the $l_i=2$ term in this case for Kleene's algorithm and matrix multiplication).
In fact, as one can see, either the span bound, or the span recurrence (\Cref{eqn:kleene-span}), does not show up in the entire analysis.

Since Kleene's algorithm is very simple, we can also show how to directly solve the recurrence by plugging \Cref{eqn:mm-par} in \Cref{eqn:kleene-cache}.
We believe this can illustrate a more intuitive idea of our analysis.
\[
\begin{split}
Q(n)=&6Q_{\smb{MM}}(n/2)+12Q_{\smb{MM}}(n/4)+\cdots+3n\cdot Q_{\smb{MM}}(1)\\
=&O\left(\frac{6n^3}{8B\sqrt{M}}+\frac{12n^3}{64B\sqrt{M}}+\cdots\right)+\\
&O\left(\frac{6P^{1/3}n^2\log^{2/3}n}{4B}+\frac{12P^{1/3}n^2\log^{2/3}n}{16B}+\cdots\right)+\\ &O\left(P\log^{2}n+2P\log^{2}(n/2)+\cdots+Pn\right)\\
=&O\left(\frac{n^3}{B\sqrt{M}}+\frac{P^{1/3}n^2\log^{2/3}n}{B}+Pn\right).
\end{split}\]
Here the terms in the first two big-Os are decreasing geometrically, while the last term increases geometrically.
The main term for parallel overhead is only a polylogarithmic factor ($O(\log^{2/3}n)$) more than the lower bound, as opposed to a polynomial factor ($O(n^{1/3})$) in the previous terms.

%Now we have shown that using \Cref{thm:recur}, it is very simple to show a tighter bound of Kleene's algorithm than previous analysis.
Although Kleene's algorithm can also be directly analyzed as shown above, using \Cref{thm:recur} enables a simpler way to show a tighter bound of Kleene's algorithm than previous analysis.
More importantly, for many algorithms that are more complicated than Kleene's algorithm, it is nearly impossible to show new bounds by directly plugging in the recurrences, in which case using \Cref{thm:recur} easily enables simple analysis and tighter bounds. We will then present these algorithms and our new analysis and bound in \Cref{sec:applications}.
%later discussed in \Cref{sec:applications}, as many of them are given in the recent paper by~\cite{BG2020}, and their analysis is much harder without the definitions and theorems in \Cref{sec:approach}.

\hide{%%%old
To start with, we use Kleene's algorithm for all-pair shortest-paths (APSP) as an example to explain the analysis.
Kleene's algorithm is very simple, so even without using \Cref{thm:recur}, it is not hard to show the parallel cache bound in \Cref{table:intro}.
However, this is not the case for many complicated algorithms later discussed in \Cref{sec:applications}, as many of them are given in the recent paper by~\cite{BG2020}, and their analysis is much harder without the definitions and theorems in \Cref{sec:approach}.

Kleene's algorithm solves the all-pair shortest-paths (APSP) problem that takes a graph $G=(V,E)$ (with no negative cycles) as input.
The Kleene's algorithm (first mentioned in~\cite{kleene1951representation,munro1971efficient,fischer1971boolean,furman1970application}, discussed in full details in~\cite{Aho74}) is a divide-and-conquer algorithm, which is I/O-efficient, cache-oblivious and highly parallelized.
The pseudocode of Kleene's algorithm is provided in Algorithm~\ref{alg:kleene}.

\begin{algorithm}[!h]
\caption{\mf{Kleene}($A$)}
\label{alg:kleene}
\DontPrintSemicolon
\fontsize{9pt}{11pt}\selectfont
\KwIn{Distance matrix $A$ initialized based on the input graph $G=(V,E)$}
\KwOut{Computed Distance matrix $A$}

\smallskip

\lIf {$|A| = 1$} {\Return{$A$}}

\smallskip

$A_{00}\gets \mf{Kleene}(A_{00})$\\
$A_{01}\gets A_{01}+A_{00}A_{01}$\\
$A_{10}\gets A_{10}+A_{10}A_{00}$\\
$A_{11}\gets A_{11}+A_{10}A_{01}$\\

\smallskip

$A_{11}\gets \mf{Kleene}(A_{11})$\\
$A_{01}\gets A_{01}+A_{01}A_{11}$\\
$A_{10}\gets A_{10}+A_{11}A_{10}$\\
$A_{00}\gets A_{00}+A_{10}A_{01}$\\

\smallskip

\Return{$A$}
\end{algorithm}

In Kleene's algorithm, the graph $G$ is represented as the matrix $A$.
$A$ is partitioned into 4 submatrices indexed as $\begin{bmatrix}A_{00}&A_{01}\\A_{10}&A_{11}\end{bmatrix}$.
The matrix multiplication is defined in a closed semi-ring with $(+, \min)$.
The high-level idea is first to compute the APSP between the first half of the vertices only using the paths between these vertices.
Then by applying some matrix multiplication, we update the shortest paths between the second half of the vertices using the computed distances from the first half.
We then apply another recursive subtask on the second half vertices.
The computed distances are finalized, and then we use them to update the shortest paths from the first-half vertices.

The cache complexity $Q(n)$ and \depth{} $D(n)$ of this algorithm follow the recursions:
\begin{equation}
%\begin{split}
Q(n)=2Q(n/2)+6Q_{\smb{MM}}(n/2)\label{eqn:kleene-cache}
%\end{split}
\end{equation}
$$D(n)=2D(n/2)+2D_{\smb{MM}}(n/2)$$
where $Q_{\smb{MM}}(n)$ is the I/O cost of a matrix multiplication of input size $n$.
Using the parallel matrix multiplication algorithm from~\cite{BG2020}, we have $D_{\smb{MM}}(n)=O(\log^2 n)$, $D(n)=O(n)$, $\displaystyle Q_{\smb{MM},1}=\Theta\left(\frac{n^3}{B\sqrt{M}}+\frac{n^2}{B}+1\right)$, and $Q_1=\displaystyle \Theta\left(\frac{n^3}{B\sqrt{M}}+\frac{n^2}{B}+n\right)$.

If we directly use the result from~\cite{Acar02}, then we get the parallel cache bound $Q_P=\displaystyle Q_1+O\left(\frac{pnM}{B}\right)$.
As the significant growth of processor count and cache size, the $O({pnM}/{B})$ term dominates unless $n$ is very large.
The tighter bound from~\cite{frigo2009cache,cole2012efficient,cole2012revisiting} shows $Q_P=\displaystyle Q_1+O\left(\frac{P^{1/3}n^{7/3}}{B}+Pn\right)$.
This bound is tighter than the previous one from~\cite{Acar02}, but the $O(P^{1/3}n^{7/3}/B)$ term still dominates unless $n=\omega(P^{1/2}M^{3/4})$, which is unlikely in practice.
The parallel lower bound for this computation~\cite{ballard2014communication,BG2020} is $\displaystyle Q_P=Q_1+\Omega\left(\frac{P^{1/3}n^2}{B}\right)$, so a polynomial gap remains between the lower and upper cache bounds.
Our analysis significantly closes this gap to polylogarithmic.

As discussed in \Cref{sec:approach}, our analysis will first use the bound from~\cite{frigo2009cache,cole2012efficient,cole2012revisiting} on the parallel matrix multiplication algorithm from~\cite{BG2020}, which is
\begin{equation}\label{eqn:mm-par}
  Q_{\smb{MM},P}=O\left(\frac{n^3}{B\sqrt{M}}+\frac{P^{1/3}n^2\log^{2/3}n}{B}+p\log^{2}n\right).
\end{equation}

Since Kleene's algorithm is very simple, we first directly solve the recurrence and show the result by plugging \Cref{eqn:mm-par} in \Cref{eqn:kleene-cache}.
\[
\begin{split}
Q(n)=&6Q_{\smb{MM}}(n/2)+12Q_{\smb{MM}}(n/4)+\cdots+3n\cdot Q_{\smb{MM}}(1)\\
=&O\left(\frac{6n^3}{8B\sqrt{M}}+\frac{12n^3}{64B\sqrt{M}}+\cdots\right)+~\\
 &O\left(\frac{6P^{1/3}n^2\log^{2/3}n}{4B}+\frac{12P^{1/3}n^2\log^{2/3}n}{16B}+\cdots\right)+\\
 &O\left(p\log^{2}n+2p\log^{2}(n/2)+\cdots+Pn\right)\\
=&O\left(\frac{n^3}{B\sqrt{M}}+\frac{P^{1/3}n^2\log^{2/3}n}{B}+Pn\right).
\end{split}\]
Here the terms in the first two big-Os are decreasing geometrically, while the last term increases geometrically.
The main term for parallel overhead is only a polylogarithmic factor ($O(\log^{2/3}n)$) more than the lower bound, as opposed to a polynomial factor ($O(n^{1/3})$) in the previous terms.

Now we show how to use \Cref{thm:recur} to directly solve this \dacrecur.
For $Q_P(n)$, we have $\alpha=\beta=2$, and $\{(k_i,l_i,m_i)\}=\{$$(1/B\sqrt{M},3,0),$$(P^{1/3}/B,2,2/3),$$(P,0,2)\}$.
Plugging in \Cref{thm:recur} directly gives the above solution.

In this case, the input size is $O(n^2)$ so the corresponding term $l_i=2>\log_\beta\alpha=1$.
Hence, even though Kleene's algorithm has linear span as opposed to the polylogarithmic span for matrix multiplication, the additional steals caused by the span will not affect the input term (the $l_i=2$ term in this case for Kleene's algorithm and matrix multiplication).
}

\subsection{Other Applications}\label{sec:applications}

We have shown the main theorem (\Cref{thm:recur}) and the intuition why it leads to better parallel cache complexity for Kleene's algorithm.
We now apply the theorem to a variety of classic cache-oblivious algorithms, which leads to better parallel cache bound.
The details of the algorithms can be found in~\cite{dinh2016extending,BG2020}.
Some of these algorithms are complicated, here we only show the recurrences and the parallel cache bounds since those are all we need.

\subsubsection{Building Blocks}

Before we go over the applications, we first show the parallel cache complexity of some basic primitives ($k$d-grid and matrix transpose) that the applications use.

\textbf{Matrix Multiplication (MM).} Matrix multiplication is modeled as a 3d-grid in \cite{BG2020}.
The sequential cache bound $Q_{\smb{MM},1}$ is $\Theta\left(\frac{n^3}{B\sqrt{M}}+\frac{n^2}{B}+1\right)$, and the parallel bound on $p$ processors $Q_{\smb{MM},P}$ is $O\left(\frac{n^3}{B\sqrt{M}}+\frac{P^{1/3}n^2\log^{2/3}n}{B}+P\log^{2}n\right)$.
Here we assume the matrix is stored in the bit-interleaved (BI) format~\cite{Frigo99}, which can be easily converted from other formats such as the row-major format with the same cost as matrix transpose.

\textbf{Matrix Transpose (MT).} Matrix transpose is another widely used primitives in cache-oblivious algorithms.
The sequential cache bound $Q_{\smb{MT},1}$ is $\Theta\left(\frac{n^2}{B}+1\right)$, and the parallel bound on $p$ processors~\cite{frigo2009cache,cole2012revisiting} is:
\begin{equation}\label{eqn:mt-par}
Q_{\smb{MT},P}=O\left(\frac{n^2}{B}+P\log^{2}n\right).
\end{equation}

\textbf{The 2d-grid.} The 2d-grid can be viewed as an analog of matrix multiplication, but is a 2 dimensional computation instead of 3 dimensional as in MM ($O(n^3)$ arithmetic operations and memory accesses).
%with one fewer dimensions.
It can also be viewed as a matrix-vector multiplication but the matrix is implicit.
The 2d-grid is a commonly used primitive in dynamic programming algorithms~\cite{BG2020}.
The sequential cache bound $Q_{\smb{2D},1}$ is $\Theta\left(\frac{n^2}{BM}+\frac{n}{B}+1\right)$, and the parallel bound on $p$ processors~\cite{frigo2009cache,cole2012revisiting} is \begin{equation}\label{eqn:2d-par}
Q_{\smb{2D},P}=O\left(\frac{n^2}{BM}+\frac{P^{1/2}n\log n}{B}+P\log^{2}n\right).
\end{equation}

\subsubsection{Gaussian Elimination}

Here we consider the parallel divide-and-conquer Gaussian elimination algorithm shown in~\cite{BG2020}, with the recurrence of the cache bound as $Q(n)=2Q(n/2)+4Q_{\smb{MM}}(n/2)$.
Compared to Kleene's algorithm (\Cref{eqn:kleene-cache}), this recurrence only differs by a constant.
Hence, the parallel cache bound is asymptotically the same as Kleene's algorithm.

\subsubsection{Triangular System Solver}

The triangular system solver (TRS) solves the linear system that takes the output of Gaussian elimination (i.e., $Ax=b$ where $A$ is an upper triangular matrix).
We consider the parallel divide-and-conquer algorithm for a triangular system solver from~\cite{BG2020} with cubic work and linear span.
The cache bound is: $$Q_{\smb{TRS}}(n)=4Q_{\smb{TRS}}(n/2)+2Q_{\smb{MM}}(n/2).$$

To analyze the parallel cache complexity, we can plug in \Cref{eqn:mm-par} and get the \dacrecur with $\alpha=4$, $\beta=2$, and $\{(k_i,l_i,m_i)\}=\{$ $(1/B\sqrt{M},3,0),$ $(P^{1/3}/B,2,2/3),$ $(P,0,2)\}$.
Applying  and \Cref{thm:recur} leads to the parallel cache complexity as
\begin{equation}\label{eqn:trs-par}
Q_{\smb{TRS},P}=O\left(\frac{n^3}{B\sqrt{M}}+\frac{P^{1/3}n^2\log^{5/3}n}{B}+Pn\right).
\end{equation}

\subsubsection{Cholesky Factorization and LU Decomposition}

Both Cholesky factorization and LU decomposition are widely used linear algebraic tools to decompose a matrix to the product of a lower triangular matrix and an upper triangular matrix.
The divide-and-conquer algorithms for Cholesky factorization and LU decomposition~\cite{dinh2016extending} are quite similar in the way that they are designed on top of triangular system solver and matrix multiplication.
The cache bounds for both algorithms are: $$Q(n)=2Q(n/2)+Q_{\smb{TRS}}(n/2)+O(1)\cdot Q_{\smb{MM}}(n/2).$$

We can plug in \Cref{eqn:mm-par} and \Cref{eqn:trs-par} to get the \dacrecur with $\alpha=2$, $\beta=2$, and $\{(k_i,l_i,m_i)\}=$ $\{(1/B\sqrt{M},3,0),$ $(P^{1/3}/B,2,5/3),$ $(P,0,2)\}$.
Since $\log_\beta\alpha=2$, the parallel bound is almost the same as $Q_{\smb{TRS}}$, except that the span for these algorithms is $O(n\log n)$, which increases the last term by a logarithmic factor.
Hence for these two problems, we have: $$Q_{p}=O\left(\frac{n^3}{B\sqrt{M}}+\frac{P^{1/3}n^2\log^{5/3}n}{B}+Pn\log n\right).$$

\subsubsection{LWS Recurrence}

The LWS (least-weighted subsequence) recurrence~\cite{hirschberg1987least} is one of the most commonly-used DP recurrences in practice.
Given a real-valued function $w(i,j)$ for integers $0\le i<j\le n$ and $D_0$, for $1\le j\le n$,
$$D_j=\min_{0\le i<j}\{D_i+w(i,j)\}.$$
This recurrence is widely used in real-world applications~\cite{kleinberg2006algorithm,knuth1981breaking,galil1992dynamic,galil1994parallel,aggarwal1990applications,kunnemann2017fine}.
Here we assume that $w(i,j)$ can be computed in constant work based on a constant size of input associated to $i$ and $j$, which is true for all these applications.

Here we consider the parallel divide-and-conquer algorithm to solve LWS recurrence from~\cite{BG2020} with quadratic work and linear span.
This algorithm partitions the problems into two halves, solves the first one, applies a 2d-grid computation, and solves the second one.
The cache bound is $Q(n)=2Q(n/2)+Q_{2D}(n/2)$.

Here, by using \Cref{eqn:2d-par}, the \dacrecur has $\alpha=2$, $\beta=2$, and $\{(k_i,l_i,m_i)\}=$ $\{(1/BM,2,0),$ $(P^{1/2}/B,1,1),$ $(P,0,2)\}$.
Since, $\log_\beta\alpha=1$, so the parallel cache bound is: $$Q_P(n)=O\left(\frac{n^2}{BM} + \frac{P^{1/2} n \log^2{n}}{B} + Pn\right).$$

\subsubsection{GAP Recurrence}

The GAP problem~\cite{galil1989speeding,galil1994parallel} is a generalization of the edit distance problem that
has many applications in molecular biology, geology, and speech recognition.
Given a source string $X$ and a target string $Y$, an ``edit'' can be a sequence of consecutive deletes corresponding to a gap in $X$, and a sequence of consecutive inserts corresponding to a gap in $Y$.
Let $w(p,q)$ $(0\le p < q\le n)$ be the cost of deleting the substring of $X$ from $(p+1)$-th to $q$-th character, $w'(p,q)$ be inserting the substring of $Y$ accordingly, and
$r(i,j)$ be the cost to change the $i$-th character in $X$ to $j$-th character in $Y$.

Let $D_{i,j}$ be the minimum cost for such transformation from the prefix of $X$ with $i$ characters to the prefix of $Y$ with $j$ characters, the recurrence for $i,j>0$ is:
\[\displaystyle D_{i,j}=\min\left\{\begin{matrix}\min_{0\le q<j}\{D_{i,q}+w'(q,j)\}\\\min_{0\le p<i}\{D_{p,j}+w(p,i)\}\\D_{i-1,j-1}+r(i,j)\end{matrix}\right.\]
corresponding to either replacing a character, inserting or deleting a substring.
The best parallel divide-and-conquer algorithm to compute the GAP recurrence is proposed by Blelloch and Gu~\cite{BG2020}.
The cache bound recurrence of the algorithm in~\cite{BG2020} is $Q(n)=4Q(n/2)+4(n/2)\cdot Q_{2D}(n/2)+2Q_{\smb{MT}}(n/2)$, which includes 4 subproblems with half size, a linear number of 2d-grid (see more details in~\cite{BG2020}), and 2 matrix transpose calls.

To derive parallel cache complexity, we can apply \Cref{eqn:mt-par} and \Cref{eqn:2d-par} and get the \dacrecur with $\alpha=4$, $\beta=2$, and $\{(k_i,l_i,m_i)\}=\{$$(1/BM,3,0),$ $(P^{1/2}/B,2,1),$ $(P,0,2)\}$.
Then using \Cref{thm:recur} gives $\log_\beta\alpha=2$ and $$Q_P(n)=O\left(\frac{n^3}{BM} + \frac{P^{1/2} n^2 \log^2{n}}{B} + Pn^{\log_23}\right).$$

\subsubsection{RNA recurrence}

The RNA problem~\cite{galil1994parallel} is a generalization of the GAP problem. %\yan{shorten}
In this problem, a weight function $w(p,q,i,j)$ is given, which is the cost to delete the substring of $X$ from $(p+1)$-th to $i$-th character and insert the substring of $Y$ from $(q+1)$-th to $j$-th character.
Similar to GAP, let $D_{i,j}$ be the minimum cost for such transformation from the prefix of $X$ with $i$ characters to the prefix of $Y$ with $j$ characters, the recurrence for $i,j>0$ is:
$$\displaystyle D_{i,j}=\min_{\substack{0\le p<i,0\le q<j}}\{D_{p,q}+w(p,q,i,j)\}.$$
%with the boundary values $D_{0,0}$, $D_{0,j}$ and $D_{i,0}$.
This recurrence is widely used in computational biology, like to compute the secondary structure of RNA~\cite{waterman1978rna}.
The RNA recurrence can be viewed as a 2d version of the LWS recurrence, and the latest algorithm from~\cite{BG2020} has the cache bound of
$Q(n)=4Q(n/2)+Q_{2D}(n^2)$.

For parallel cache complexity, the \dacrecur by plugging in \Cref{eqn:2d-par} is $\alpha=4$, $\beta=2$, and $\{(k_i,l_i,m_i)\}=$ $\{(1/BM,4,0),$ $(P^{1/2}/B,2,1),$ $(P,0,2)\}$.
The parallel cache bound can be solved as: $$Q_P(n)=O\left(\frac{n^4}{BM} + \frac{P^{1/2} n^2 \log^2{n}}{B} + Pn^{\log_23}\right).$$

\subsubsection{Protein Accordion Folding}

The recurrence for protein accordion folding~\cite{tithi2015high} is: $$D_{i,j}=\max_{1\le k< j-1}\{D_{j-1,k}+w(i,j,k)\}$$ for $1\le j<i\le n$, with $O(n^2/B)$ cost to precompute $w(i,j,k)$.
In~\cite{BG2020}, a parallel divide-and-conquer algorithm for protein accordion folding is given, with cubic work and linear span.
The recurrence for the cache bound is $Q(n)=2Q(n/2)+Q_{MT}(n/2)+(n/2)\cdot Q_{2D}(n/2)$, which includes 2 subproblems with half size, a linear number of 2d-grid, and 1 matrix transpose call.

For parallel cache complexity, we can apply \Cref{eqn:mt-par} and \Cref{eqn:2d-par} and get $\alpha=2$, $\beta=2$, and $\{(k_i,l_i,m_i)\}=\{$$(1/BM,3,0),$ $(P^{1/2}/B,2,1),$ $(P,0,2)\}$.
Since $\log_\beta\alpha=1$, we can get:
$$Q_P(n)=O\left(\frac{n^3}{BM} + \frac{P^{1/2} n^2 \log{n}}{B} + Pn\log^2n\right).$$

\subsubsection{Parenthesis Recurrence}

The Parenthesis recurrence solves the following problem in many applications~\cite{CLRS,galil1992dynamic,galil1994parallel,yao1980efficient}: for a linear sequence of objects, an associative binary operation, and the cost of performing that operation on any two objects, the goal is to compute the min-cost way to group the objects by repeatedly applying the binary operations.
Let $D_{i,j}$ be the minimum cost to merge the objects indexed from $i+1$ to $j$ (1-based), the recurrence for $0\le i<j\le n$ is:
$$D_{i,j}=\min_{i<k<j}\{D_{i,k}+D_{k,j}+w(i,k,j)\}$$ where $w(i,k,j)$ is the cost to merge the two partial results of objects indexed from $i+1$ to $k$ and those from $k+1$ to $j$.

The parallel cache-oblivious algorithm for the parenthesis
recurrence is introduced in~\cite{chowdhury2008cache}.
The original problem and some subproblems are triangle ones ($Q_\triangle$), but when computing them, we need other square subproblems ($Q_\square$).
The recurrence relations for the cache bounds are $Q_\triangle(n)=2Q_\triangle(n/2)+Q_\square(n/2)$,
and $Q_\square(n)=4Q_\square(n/2)+4Q_{\smb{MM}}(n/2)$.

To compute the parallel cache complexity, we apply \Cref{eqn:mm-par} and \Cref{thm:recur} to $Q_\square$ first to obtain a parallel cache complexity of $Q_{\square,P}(n)= O\left(\frac{n^3}{B\sqrt{M}} + \frac{P^{1/3}n^2 \log^{5/3}{n}}{B} + Pn^{\log_23}\right)$. Then we substitute this into $Q_\triangle$ to apply \Cref{thm:recur} again and get:
$$Q_P(n)=Q_{\triangle,P}(n)=O\left(\frac{n^3}{B\sqrt{M}} + \frac{P^{1/3}n^2 \log^{5/3}{n}}{B} + Pn^{\log_23}\right).$$

\subsection{Discussions}\label{sec:discuss}

A common theme in our analysis is to apply the parallel cache bounds of the basic primitives from~\cite{frigo2009cache} in a certain step, and then use the recursive structure to derive parallel cache bound for the entire algorithm.
We note that in many cases this is better than directly using the result in~\cite{frigo2009cache} for the entire algorithm.
However, we note that this is not always true, and it relies on the recursive structure.
More precisely, it depends on whether the number of base-case subproblems is asymptotically no more than the input elements (true for all algorithms discussed above).
A counterexample is the edit distance problem (or longest common subsequence).
The recurrence for the divide-and-conquer algorithm is $Q(n)=4Q(n/2)+O(1)$.
In this case, using \Cref{thm:recur} gives a looser bound than using the analysis from~\cite{frigo2009cache}.
Hence, how to tightly bound the parallel cache complexity for edit distance remains an open problem.

\hide{
\section*{Acknowledgement}
This work is in partial supported by NSF grant CCF-2103483.
}
\balance

%\bibliographystyle{abbrv}
%\bibliography{bib/strings,bib/main}
%\input{sssp.bbl}

\end{document}